\documentclass[12pt]{article}
\usepackage{amsmath}
\usepackage{amssymb,amscd}
\usepackage{bm}
\usepackage{wrapfig}

\usepackage{color}
\usepackage{bm}
\usepackage[all]{xy}

\usepackage{theorem}

\newtheorem{theorem}{Theorem}[section]

{\theorembodyfont{\normalfont}
 \newtheorem{definition}[theorem]{Definition}
 \newtheorem{example}{Example}[section]
 
 }

\newenvironment{proof}[1][Proof]{\begin{trivlist}
\item[\hskip \labelsep {\bfseries #1}]}{\end{trivlist}}
\newcommand{\qed}{\nobreak \ifvmode \relax \else
      \ifdim\lastskip<1.5em \hskip-\lastskip
      \hskip1.5em plus0em minus0.5em \fi \nobreak
      \vrule height0.75em width0.5em depth0.25em\fi}


\newlength{\minitwocolumn}
\usepackage{amsmath}
\usepackage{amssymb,amscd}
\usepackage{bm}
\usepackage{amsbsy} 
\usepackage{wrapfig}
\usepackage{theorem}

\newcommand{\beq}{\begin{equation}}
\newcommand{\eeq}{\end{equation}}
\newcommand{\bea}{\begin{eqnarray*}}
\newcommand{\eea}{\end{eqnarray*}}
\newcommand{\beqa}{\begin{eqnarray}}
\newcommand{\eeqa}{\end{eqnarray}}

\def\bR{{\mathbb{R}}}

\def\bZ{{\mathbb{Z}}}

\newcommand{\calH}{{\mathcal H}}

\newcommand{\calL}{{\mathcal L}}
\newcommand{\calM}{{\mathcal M}}

\newcommand{\calO}{{\mathcal O}}

\newcommand{\calS}{{\mathcal S}}

\newcommand{\sbv}[2]{{\{{{#1},{#2}}\}}}

\newcommand{\courant}[2]{{[{{#1},{#2}}]_D}}

\newcommand{\bracket}[2]{\langle #1,\,#2\rangle}

\newcommand{\inner}[2]{{({{#1},{#2}})}}

\newcommand{\rd}{\mathrm{d}}

\def\xzero{X{}}

\newcommand{\ue}{{\underline{e}}{}}

\begin{document}


\baselineskip 0.7cm

\begin{titlepage}
\begin{flushright}
\end{flushright}

\vskip 1.35cm
\begin{center}
{\Large \bf
Higher dimensional Lie algebroid sigma model with WZ term
}
\vskip 1.2cm
Noriaki Ikeda
\footnote{E-mail:\
nikedaATse.ritsumei.ac.jp
}
\vskip 0.4cm

{\it
Department of Mathematical Sciences,
Ritsumeikan University \\
Kusatsu, Shiga 525-8577, Japan \\
}
\vskip 0.4cm

\today

\vskip 1.5cm

\begin{abstract}
We generalize the $(n+1)$-dimensional twisted $R$-Poisson topological 
sigma model with flux on a target Poisson manifold to a Lie algebroid.
Analyzing consistency of constraints in the Hamiltonian formalism and the gauge symmetry in the Lagrangian formalism, geometric conditions of the target space to make the topological sigma model consistent are identified.
The geometric condition is an universal compatibility condition of a Lie algebroid with the multi-symplectic structure. This condition is a generalization of the momentum map theory of a Lie group and is regarded as a generalization of the momentum section condition of the Lie algebroid.
\end{abstract}
\end{center}
\end{titlepage}

\tableofcontents

\setcounter{page}{2}


\rm

\section{Introduction}
\noindent
Algebroid structures appear as background mathematical structures in physics, such as T-duality in string theory \cite{Grana:2008yw, 
Cavalcanti:2011wu, Blumenhagen:2012pc, Asakawa:2014kua, Severa:2015hta, Heller:2016abk},
gauged nonlinear sigma models
\cite{Chatzistavrakidis:2015lga, Chatzistavrakidis:2016jfz, Chatzistavrakidis:2016jci, Chatzistavrakidis:2017tpk, Bouwknegt:2017xfi, Bugden:2018pzv, Wright:2019pru, Ikeda:2019pef},
topological sigma models \cite{Ikeda:2012pv, Chatzistavrakidis:2019rpp, Grewcoe:2020uih, Marotta:2021sia}, 
double field theory
\cite{Siegel:1993th, Siegel:1993xq, Hull:2009mi, Hull:2009zb, Chatzistavrakidis:2018ztm, Chatzistavrakidis:2019huz, Grewcoe:2020hyo}, etc.
It is important to analyze geometric structures in duality physics.
A Lie algebroid \cite{Mackenzie}, which is a generalization of a Lie algebra, 
is the most fundamental algebroid structure. In this paper, we propose a new topological sigma model with a Lie algebroid structure.

The Poisson structure is not only a fundamental structure of the classical mechanics but also a generalization of a Lie algebra, which mainly appears as symmetries.
It is defined by a bivector field $\pi \in \Gamma(\wedge^2 TM)$ satisfying
$[\pi, \pi]_S = 0,$
where $[-,-]_S$ is the Schouten bracket defined on the space of multivector fields $\Gamma(\wedge^{\bullet} TM)$.
A sigma model with the Poisson structure, 
the Poisson sigma model \cite{Ikeda:1993fh, Schaller:1994es},
describes topological aspects of the NS-flux and
has many applications such as the derivation of Kontsevich formula 
of the deformation quantization \cite{Cattaneo:1999fm}.
The Poisson sigma model is generalized to the twisted Poisson sigma model by introducing the WZ term as a consistent constrained mechanical system. Consistency requires the deformation of the Poisson structure to the twisted Poisson structure \cite{Klimcik:2001vg, Park:2000au, Severa:2001qm}.
The twisted Poisson structure is defined by equations,
\begin{eqnarray}
&& \frac{1}{2}[\pi, \pi]_S = \bracket{\otimes^{3} \pi}{H},
\\
&& \rd H =0,
\end{eqnarray}
where $H$ is a closed $3$-form.
For a manifold $M$ with a Poisson or a twisted Poisson structure, the cotangent bundle $T^*M$ has a Lie algebroid structure.
Thus, it is interesting to generalize a Poisson or a twisted Poisson structure to a general Lie algebroid case.

Recently, Chatzistavrakidis has proposed a higher generalization of the twisted Poisson structure and the twisted Poisson sigma model by considering a higher dimensional topological sigma model \cite{Chatzistavrakidis:2021nom}.
It is a topological sigma model with WZ term on a $(n+1)$-dimensional worldvolume.  The twisted $R$-Poisson structure is defined by the following condition,
\begin{eqnarray}
&& [\pi, \pi]_S = 0,
\label{RPoisson1}
\\
&& [\pi, J]_S = \frac{}{} 
\bracket{\otimes^{n+2} \pi}{H},
\label{RPoisson2}
\\
&& \rd H =0,
\label{RPoisson3}
\end{eqnarray}
where $\pi$ is the Poisson bivector field, 
$H \in \Omega^{n+2}(M)$ is a closed $(n+2)$-form and 
$J \in \Gamma(\wedge^{n+1}(M))$ is an $(n+1)$-multivector field on $M$.
\footnote{In this paper, we denote a multivector field by $J$ though
it is denoted by $R$ in the paper \cite{Chatzistavrakidis:2021nom}. 
$R$ is used for a curvature.}

In this paper, we consider a new topological sigma model by generalizing 
the Poisson structure to a Lie algebroid in the twisted $R$-Poisson sigma 
model.
The key equation is
\begin{eqnarray}
&& {}^E \rd J = - \bracket{\otimes^{n+2} \rho}{H},
\label{momsec21}
\end{eqnarray}
where ${}^E \rd $ is the Lie algebroid differential, 
$J$ is an $E$-$(n+1)$-form, 
$\rho$ is the so called anchor map of a Lie algebroid and $H$ is a closed $(n+2)$-form.
We analyze mathematical structures of Equation \eqref{momsec21} in details
in Section \ref{LAmultisym}. We show that the total structure is regarded as 
a higher Dirac structure of a Lie $(n+1)$-algebroid.

Another purpose is to generalize the so called AKSZ sigma models
\cite{Alexandrov:1995kv, Cattaneo:2001ys, Ikeda:2001fq, Roytenberg:2006qz}
adding the WZ term.
The AKSZ construction of topological sigma models is a clear geometric construction method of the rather complicated BFV formalism 
\cite{Batalin:1977pb, Batalin:1983pz} and the BV formalism 
\cite{Batalin:1981jr, Batalin:1984jr}
from a classical action based on graded symplectic geometry. The BV bracket and the BV action are directly constructed by pullbacks of the target space graded symplectic structure.
For instance, refer to 
a review of AKSZ sigma models \cite{Ikeda:2012pv}. 
However, the AKSZ construction does not work if we twist the classical action adding the WZ term.
In two dimensional case, the BV and BFV formalisms of the twisted Poisson sigma model have been constructed in the paper \cite{Ikeda:2019czt}, and it was discussed that the correct BV action of the twisted PSM was not obtained by the genuine AKSZ procedure.
In order to consider generalizations to higher dimensions,
first we need to clarify background geometric structures of 
higher dimensional twisted topological sigma models with the WZ term.

This paper is organized as follows.
In Section 2, we introduce a topological sigma model 
with a Lie algebroid structure and WZ term.
In Section 3, we prepare geometric structures which appear in our model 
such as a Lie algebroid, a pre-multisymplectic structure and their compatibility condition. We also explain some related examples.
In Section 4, we analyze the Hamiltonian formalism and show that the theory is consistent if and only if the geometric compatibility condition holds.
In Section 5, the Hamiltonian formalism is rewritten to the target space covariant expression. All equations are described by geometric quantities of the target manifold.
In Section 6, we consider the Lagrangian formalism and obtain consistent gauge transformations under the same geometric compatibility condition.
In Section 7, we rewrite gauge transformations to the manifestly covariant formulation.
Section 8 is devoted to discussion and outlook.
In Appendix \ref{geometryofLA}, some formulas are summarized.

\section{Lie algebroid topological sigma model with flux and WZ term}
\label{LASMFWZ}
Let $N$ be an $n+2$ dimensional manifold with $n+1$ dimensional boundary, 
$\Sigma = \partial N$. 
Consider a $d$-dimensional target space $M$ and a vector bundle $E$ over $M$.
Suppose $E$ has a Lie algebroid structure.
A Lie algebroid has two operations, a Lie bracket $[-,-]$ on $\Gamma(E)$
and the bundle map $\rho:E \rightarrow TM$ called the anchor map.
A Lie algebroid is reviewed in Section \ref{LAmultisym}.
We introduce the pairing of $TM$ and $T^*M$, $\bracket{-}{-}$,
and the pairing of $E$ and $E^*$, $\inner{-}{-}$.

We consider a smooth map from $N$ to $M$, $X: N \rightarrow M$.
$A \in \Gamma(T^*\Sigma, X^* E)$ is a $1$-form taking a value on the pullback of $E$, $X^* E$.
$Y \in \Gamma(\wedge^{n-1} T^*\Sigma, X^* E^*)$ is an $(n-1)$-form taking a value on $X^* E^*$.
$Z \in \Gamma(\wedge^{n} T^*\Sigma, X^* T^*M)$ is an $n$-form taking a value on $X^* T^*M$.
We consider the following sigma model action functional,
\begin{eqnarray}
S &=& \int_{\Sigma} 
\left[\bracket{Z}{\rd X} + \inner{Y}{\rd A}
- \bracket{Z}{X^* \rho (A)}
+ \frac{1}{2} \inner{Y}{X^* [A, A]}
+ X^* J(A, \ldots, A)
\right]
\nonumber \\ &&
+ \int_{N} X^* H.
\label{classicalactionofHLASM0}
\end{eqnarray}
Here $\rd$ is the de Rham differential on $\Sigma$.
For pairings of pullbacks by $X$, the same notation are used, i.e.,
$\bracket{-}{-}$ is the pairing of a pullback of $TM$ and $T^*M$,
and
$\inner{-}{-}$ is the pairing of a pullback of $E$ and $E^*$.
$J \in \Gamma(\wedge^{n+1}E^*)$ is an $E$-$(n+1)$-form on $E$
and $H \in \Omega^{n+2}(M)$ is an $(n+2)$-form on $M$.

Taking local coordinates on $M$ and $E$, we have four kinds of fields $X^i$, $Z_i$, $A^a$ and $Y_a$, 
where $i$ is the index of $M$ and $a$ is the index of the fiber of $E$.
The action is
\begin{align}
S &= \int_{\Sigma} 
\left[Z_i \wedge \rd X^i + Y_a \wedge \rd A^a 
- \rho^i_a(\xzero) Z_i \wedge A^a
+ \frac{1}{2} C_{ab}^c(\xzero) Y_c \wedge A^a \wedge A^b
\right.
\nonumber \\ &
\left.
+ \frac{1}{(n+1)!} J_{a_1\ldots a_{n+1}}(\xzero) A^{a_1} \wedge \ldots \wedge 
A^{a_{n+1}}
\right]
\nonumber \\ &
+ \int_{N} \frac{1}{(n+2)!} H_{i_1\ldots i_{n+2}}(\xzero) 
\rd X^{i_1} \wedge \ldots \wedge \rd X^{i_{n+2}}.
\label{classicalactionofHLASM}
\end{align}
$\rho^i_a$ is local coordinate expression of the anchor map $\rho$, 
$C_{ab}^c$ are the structure functions of the Lie bracket, 
$J_{a_1\ldots a_{n+1}}$ and $H_{i_1\ldots i_{n+2}}$ 
are $J$ and $H$, which are completely antisymmetric tensors.
We call this model the twisted Lie algebroid sigma model with flux,
or the Lie algebroid sigma model with the WZ term.

The equations of motion are computed as
\begin{align}
F^i_X &:= \rd X^i - \rho^i_a(X) A^a = 0  \, ,
 \label{eom1}\\
F_{A}^a &:= \rd A^a + \tfrac{1}{2} C^a_{bc}(X) A^b \wedge A^c = 0  \, ,
\label{eom2}
\\
F_{Ya} &:= \rd Y_a + (-1)^{n} \rho^i_a Z_i
+ (-1)^{n-1} C_{ab}^c Y_c \wedge A^b + \frac{1}{n!} J_{a b_2\ldots b_{n+1}}(\xzero) A^{b_2} \wedge \ldots \wedge 
A^{b_{n+1}}  = 0  \, ,
 \label{eom3}\\
F_{Zi} &:= (-1)^{n} \rd Z_i - \partial_i \rho^j_a Z_j \wedge A^a 
+ \tfrac{1}{2} \partial_i C^a_{bc}
Y_a \wedge A^b \wedge A^c
+ \frac{1}{(n+1)!} \partial_i J_{a_1\ldots a_{n+1}}(\xzero) A^{a_1} \wedge \ldots \wedge A^{a_{n+1}}
\nonumber \\
&+ \frac{1}{(n+1)!} H_{ij_1\ldots j_{n+1}}
\rd X^{j_1} \wedge \ldots \wedge \rd X^{j_{n+1}} = 0  \, .
 \label{eom4}
\end{align}

\section{Lie algebroid and compatible $E$-flux on pre-multisymplectic manifold}
\label{LAmultisym}
\noindent
In this section, we explain the background geometry of the sigma model 
\eqref{classicalactionofHLASM0} introduced in Section \ref{LASMFWZ}.

\subsection{Lie algebroid}
Since we want to consider a generalization of the $R$-Poisson structure, 
we assume that the target vector bundle is a Lie algebroid.
\begin{definition}
Let $E$ be a vector bundle over a smooth manifold $M$.
A Lie algebroid $(E, \rho, [-,-])$ is a vector bundle $E$ with
a bundle map $\rho: E \rightarrow TM$ called the anchor map, 
and a Lie bracket
$[-,-]: \Gamma(E) \times \Gamma(E) \rightarrow \Gamma(E)$
satisfying the Leibniz rule,
\begin{eqnarray}
[e_1, fe_2] &=& f [e_1, e_2] + \rho(e_1) f \cdot e_2,
\end{eqnarray}
{where $e_i \in \Gamma(E)$ and $f \in C^{\infty}(M)$.}
\end{definition}
Local coordinate expressions of formulas in a Lie algebroid 
are listed in Appendix \ref{geometryofLA}.

A Lie algebroid is a generalization of a Lie algebra and the space of vector fields.
\begin{example}
Let a manifold $M$ be one point $M = \{pt \}$. 
Then a Lie algebroid is a Lie algebra $\mathfrak{g}$.
\end{example}
\begin{example}
If a vector bundle $E$ is a tangent bundle $TM$ and $\rho = \mathrm{id}$, 
then a bracket $[-,-]$ is a normal Lie bracket of vector fields
and $(TM, \mathrm{id}, [-,-])$ is a Lie algebroid.
\end{example}
\begin{example}\label{actionLA}
Let $\mathfrak{g}$ be a Lie algebra and assume an infinitesimal action of 
$\mathfrak{g}$ on a manifold $M$.
$\mathfrak{g}$ acts as a differential operator,
the infinitesimal action 
determines a map $\rho: M \times \mathfrak{g} \rightarrow TM$.
The consistency of a Lie bracket requires a Lie algebroid structure on 
$(E= M  \times \mathfrak{g}, \rho, [-,-])$.
This Lie algebroid is called an action Lie algebroid.
\end{example}

\begin{example}\label{Poisson}
An important nontrivial Lie algebroid is a Lie algebroid induced from a Poisson structure. A bivector field $\pi \in \Gamma(\wedge^2 TM)$ is called a Poisson structure if $[\pi, \pi]_S =0$, where $[-,-]_S$ is a Schouten bracket on $\Gamma(\wedge^{\bullet} TM)$.

Let $(M, \pi)$ be a Poisson manifold. Then, we can define a bundle map,
$\pi^{\sharp}: T^*M \rightarrow TM$ by $\pi^{\sharp}(\alpha)(\beta) = \pi(\alpha, \beta)$ for all $\beta \in \Omega^1(M)$.
A Lie bracket on $\Omega^1(M)$ is defined by the so called Koszul bracket,
\begin{eqnarray}
[\alpha, \beta]_{\pi} = L_{\pi^{\sharp} (\alpha)}\beta - L_{\pi^{\sharp} (\beta)} \alpha - \rd(\pi(\alpha, \beta)),
\end{eqnarray}
where $\alpha, \beta \in \Omega^1(M)$.
Then, $(T^*M, -\pi^{\sharp}, [-, -]_{\pi})$ is a Lie algebroid.
\end{example}

\begin{example}\label{tPoisson}
More generally, Let $(M, \pi, H)$ be a twisted Poisson manifold.
i.e., suppose that a bivector field $\pi \in \Gamma(\wedge^2 TM)$ and 
$H \in \Omega^3(M)$ satisfy  the following equations:
\begin{eqnarray}
&& \frac{1}{2}[\pi, \pi]_S = \bracket{\otimes^{3} \pi}{H},
\label{tPoisson1}
\\
&& \rd H =0,
\end{eqnarray}

If we define a bundle map,$\pi^{\sharp}: T^*M \rightarrow TM$ 
and a Lie bracket on $\Omega^1(M)$,
\begin{eqnarray}
[\alpha, \beta]_{\pi,H} = L_{\pi^{\sharp} (\alpha)}\beta - L_{\pi^{\sharp} (\beta)} \alpha - \rd(\pi(\alpha, \beta))
+ \iota_{\alpha} \iota_{\beta} H,
\end{eqnarray}
for $\alpha, \beta \in \Omega^1(M)$.
Then, $(T^*M, -\pi^{\sharp}, [-, -]_{\pi, H})$ is a Lie algebroid.
\end{example}

One can refer to many other examples, for instance, in \cite{Mackenzie}.

\subsection{Lie algebroid differential}
Consider the spaces of exterior products of sections of $E^*$ called 
the space of $E$-differential forms, $\Gamma(\wedge^{\bullet} E^*)$.
We define a Lie algebroid differential ${}^E \rd: \Gamma(\wedge^m E^*)
\rightarrow \Gamma(\wedge^{m+1} E^*)$
such that $({}^E \rd)^2=0$. 
\begin{definition}
A Lie algebroid differential ${}^E \rd: \Gamma(\wedge^m E^*)
\rightarrow \Gamma(\wedge^{m+1} E^*)$ is defined by
\begin{eqnarray}
{}^E \rd \alpha(e_1, \ldots, e_{m+1}) 
&=& \sum_{i=1}^{m+1} (-1)^{i-1} \rho(e_i) \alpha(e_1, \ldots, 
\check{e_i}, \ldots, e_{m+1})
\nonumber \\ && 
+ \sum_{1 \leq i < j \leq m+1} (-1)^{i+j} \alpha([e_i, e_j], e_1, \ldots, \check{e_i}, \ldots, \check{e_j}, \ldots, e_{m+1}),
\label{LAdifferential}
\end{eqnarray}
where $\alpha \in \Gamma(\wedge^m E^*)$ and $e_i \in \Gamma(E)$.
\footnote{In Equation \eqref{LAdifferential}, indices $i, j$ are not indices of local coordinates on $M$, but counting of elements of $\Gamma(E)$.}
\end{definition}
One can easily check $({}^E \rd)^2=0$  using identities of the Lie algebroid.

Lie algebroids are described by means of $\mathbb{Z}$-graded geometry \cite{Vaintrob}. A graded manifold $E[1]$ for a vector bundle $E$ are
shifted vector bundle spanned by local coordinates $x^i, \ (i=1, \ldots, \mathrm{dim} M)$ on the base manifold $M$ of degree zero, and $q^a, \ (a=1, \ldots, \mathrm{rank} E)$ on the fiber of degree one, respectively. Degree one coordinate $q^a$ has the property, $q^a q^b = - q^b q^a$. $E$-differential forms which are 
sections of $\wedge^{\bullet} E^*$ are identified with functions on the graded manifold $E[1]$, i.e., $C^{\infty}(E[1]) \simeq \Gamma(\wedge^{\bullet} E^*)$,
where the degree one odd coordinate $q^a$ is identified by a basis $e^a$ of sections of $E^*$.
A product for homogeneous elements $f, g \in C^{\infty}(E[1])$ has the
property, $fg = (-1)^{|f||g|} gf$, where $|f|$ is degree of $f$.
The differential operator of degree $-1$, $\frac{\partial}{\partial q^a}$, is the derivation satisfying $\frac{\partial}{\partial q^a} q^b = \delta^b_a$, which is a linear operator on a space of functions satisfying the Leibniz rule.  

We define a degree plus one vector field $Q$ on $E[1]$:
 \beq \label{Q}
 Q = \rho_a^i (x) q^a \frac{\partial}{\partial x^i} - \frac{1}{2} C^c_{ab}(x) q^a q^b  \frac{\partial}{\partial q^c} \, ,
 \eeq
Then, the odd vector field $Q$ satisfies 
 \beq  Q^2 = 0 \, .\label{Q2}
 \eeq
if and only if $\rho, C$ are the anchor map and the structure function of a Lie algebroid on $E$.
Identifying functions on $C^{\infty}(E[1]) \simeq \Gamma(\wedge^\bullet E^*)$, 
$Q$ is the Lie algebroid differential ${}^E \rd$.

We explain the precise correspondence of $Q$ with ${}^E\rd$.
For $e^a$, the basis of $E^*$,
the map 
$j: \Gamma(\wedge^{\bullet}E^*) \rightarrow C^{\infty}(E[1])$
is given by the map of basis,
$j:(x^i, e^a) \mapsto (x^i, q^a)$.
The differential ${}^E\rd$ on $\Gamma(\wedge^{\bullet}E^*)$ is defined by the pullback, ${}^E\rd=j^*Q$.

\subsection{Compatible condition of $E$-differential form with 
pre-multisymplectic form}
We introduce another geometric notion which appears in the topological 
sigma model \eqref{classicalactionofHLASM0}. 
It is a condition on a pre-multisymplectic structure analogous to
the condition of the momentum map in the symplectic manifold.
\begin{definition}
A pre-$(n+1)$-plectic form $H$ is a closed $(n+2)$-form on a 
smooth manifold $M$, i.e., $\rd H=0$.
A manifold $M$ with a pre-$(n+1)$-plectic form $H$ is called a pre-$(n+1)$-plectic manifold.
\end{definition}
A pre-$(n+1)$-plectic manifold is also called a pre-multisymplectic manifold for $n \geq 1$.
A pre-$(n+1)$-plectic structure is called an $(n+1)$-plectic structure if 
$H$ is nondegenerate, i.e., if $\iota_{v} H =0$ for a vector field $v \in \mathfrak{X}(M)$ is equivalent to $v=0$.
A $1$-plectic manifold ($n=0$) is nothing but a symplectic manifold.

We introduce an ordinary connection $\nabla$ on the vector bundle $E$. i.e.,
a covariant derivative $\nabla: \Gamma(E) \rightarrow \Gamma(E \otimes T^*M)$, satisfying $\nabla(fe) =f \nabla e +  \nabla f \otimes e$ for a section $e \in \Gamma(E)$ and a function $f \in C^{\infty}(M)$. 
A dual connection on $E^*$ is defined by 
\begin{eqnarray}
\rd \inner{\mu}{e} = \inner{\nabla \mu}{e} + \inner{\mu}{\nabla e},
\end{eqnarray}
for all sections $\mu \in \Gamma(E^*)$ and $e \in \Gamma(E)$.
The connection is extended to the space of differential forms
and the dual connection extends to a degree 1 operator on the space of differential forms $\Omega^k(M, E)$ and $\Omega^k(M, E^*)$ .

An \textit{$E$-connection} ${}^E \nabla: \Gamma(TM) \rightarrow \Gamma(TM \otimes E^*) $ on the space of sections $\Gamma(TM)$ is defined by
\begin{eqnarray}
{}^E \nabla_{e} v &:=& \calL_{\rho(e)} v + \rho(\nabla_v e)
= [\rho(e), v] + \rho(\nabla_v e),
\end{eqnarray}
where 
$e \in \Gamma(E)$ and $v \in \Gamma(TM)$.

For an $(n+2)$-form $H$ and the anchor map $\rho$,
$\bracket{\otimes^{n+2} \rho}{H}$ is defined by
\begin{eqnarray}
&& 
\bracket{\otimes^{n+2} \rho}{H}(e_1, \ldots, e_{n+2})
= (\iota_{\rho})^{(n+2)} H(e_1, \ldots, e_{n+2})
= H (\rho(e_1), \ldots, \rho(e_{n+2})),
\end{eqnarray}
for $e_i \in \Gamma(E)$.

We introduce a new notion.
\begin{definition}\label{bracomp}
Let $(M, H)$ be a pre-$(n+1)$-plectic manifold
and $(E, \rho, [-,-])$ be a Lie algebroid over $M$.
Then, an $E$-$(n+1)$-form $J \in \Gamma(\wedge^{n+1} E^*)$ is called 
\textit{bracket-compatible} if $J$ satisfies
\begin{eqnarray}
&& {}^E \rd J
= - \bracket{\otimes^{n+2} \rho}{H}
= 
- (\iota_{\rho})^{n+2} H.
\label{momsec1}
\end{eqnarray}
Then, a flux $J$ is also called \textit{compatible} with a pre-multisymplectic form $H$.
\end{definition}
The important note is the left hand side 
in \eqref{momsec1} is the $E$-derivative 
${}^E \rd$, not the $E$-covariant derivative ${}^E \nabla$.

The condition \eqref{momsec1} appears in many situations
as we list up some examples below.
This condition is regarded as one universal generalization
of compatibility conditions of a Lie algebroid structure
with a pre-multisymplectic form.

Some known geometric structures are regarded as special cases of 
Equation \eqref{momsec1}.
\begin{example}[Twisted Poisson structure]
Let $(\pi, H)$ be a twisted Poisson structure on $M$.
In this case, the cotangent bundle $T^*M$ has a Lie algebroid structure
as explained in Example \ref{tPoisson}.
Using the Lie algebroid differential ${}^E \rd$ induced from 
this Lie algebroid, Equation \eqref{tPoisson1} is rewritten as
\begin{eqnarray}
&& {}^E \rd \pi = - \bracket{\otimes^{3} \pi}{H}.
\end{eqnarray}
$J= \pi$ is bracket-compatible on a pre-$2$-plectic manifold
with a pre-$2$-plectic form $H$.
\end{example}

\begin{example}[twisted $R$-Poisson structure]
Let $M$ be a twisted $R$-Poisson manifold. \cite{Chatzistavrakidis:2021nom}
$\pi \in \Gamma(\wedge^2 TM)$ is a Poisson bivector field,
 $H$ is a closed $(n+2)$-form, and $J \in \Gamma(\wedge^{n+1}TM)$ is an $(n+1)$-multivector field.
As explained in Example \ref{Poisson},
the Poisson bivector field $\pi$ induces a Lie algebroid structure
on $T^*M$.
Under this Lie algebroid structure, the only nontrivial condition of $R$-Poisson structure \eqref{RPoisson2} is written as
\begin{eqnarray}
&& {}^E \rd J = \bracket{\otimes^{n+2} \pi}{H}.
\end{eqnarray}
$- J$ is bracket-compatible for the pre-$(n+2)$-plectic form $H$.
\end{example}

\begin{example}[Momentum section]
The terminology '\textit{bracket-compatible}' comes from the momentum section theory with a Lie algebroid action on a symplectic manifold, which is a generalization of the moment map theory on a symplectic manifold with a Lie group (Lie algebra) action. \cite{Blohmann:2018} See also \cite{Kotov:2016lpx, Ikeda:2019pef, Ikeda:2021fjk}.

Suppose that a base manifold $M$ is a pre-symplectic manifold, i.e., $M$ has a a closed $2$-form $\omega = H \in \Omega^2(M)$, which is not necessarily nondegenerate. Moreover, suppose a Lie algebroid $(E, \rho, [-,-])$ over $M$.

\begin{definition} 
A section $\mu \in \Gamma(E^*)$ of $E^*$ is called a momentum section 
if $\mu$ satisfies the following two conditions.\footnote{
The connection is denoted by $D$ in the papers \cite{Blohmann:2018}
and \cite{Ikeda:2019pef}.}
\medskip\\
\noindent
(M1)
A section $\mu \in \Gamma(E^*)$ is a \emph{momentum section} if
\begin{eqnarray}
\nabla \mu = - \iota_{\rho} \omega.
\label{HH2}
\end{eqnarray}
%
(M2)
A momentum section $\mu$ is \emph{bracket-compatible} if
\begin{eqnarray}
{}^E \rd \mu 
= - \bracket{\rho^{\otimes 2}}{\omega}
= - \iota_{\rho}^{2}{\omega}.
\label{HH3}
\end{eqnarray}
\end{definition}
For an action Lie algebroid $E=M \times \mathfrak{g}$,
a momentum section reduces a momentum map.
Since we can take the zero connection $\nabla=d$ for the trivial bundle,
the condition (M1) is $\rd \mu = - \iota_{\rho} \omega$.
The condition (M2) reduces to the equivariant condition,
\begin{eqnarray}
\mathrm{ad}^*_{e_1} \mu(e_2) = \mu([e_1, e_2]).
\label{momentmap3}
\end{eqnarray}
for $e_1, e_2 \in \mathfrak{g}$ using \eqref{HH2}.

If we take $n=0$ in Definition \ref{bracomp} and $J= \mu$, Equation \eqref{momsec1} coincides with the condition (M2). Therefore Equation \eqref{momsec1} is a generalization of the bracket-compatible condition of the momentum section to a pre-multisymplectic manifold.

We make several comments about relations with our theory to 
the above definitions of momentum sections.
The condition corresponding to (M1), Equation \eqref{HH2}, does not appear in our model.
It is because our model is purely a topological sigma model. 
Refer to \cite{Ikeda:2019pef} about relations of the conditions (M1) and (M2) with the Hamiltonian mechanics.
Since the Hamiltonian is zero, $\calH=0$,
we obtain only the consistency conditions of constraints,
which is identified to the condition (M2).
The condition (M1) is related to consistency with the Hamiltonian and constraints as discussed in \cite{Ikeda:2019pef}. 
If we consider non topological gauged nonlinear sigma models, the condition (M1) is needed as the consistency condition of gauge invariance. 

The following one more condition (M0) is imposed in the paper \cite{Blohmann:2018}.
\medskip\\
\noindent
(M0)
$E$ is \emph{presymplectically anchored with respect to $\nabla$} if
\begin{eqnarray}
\nabla^2 \mu = 0,
\label{HH1}
\end{eqnarray}
The condition (M0) is regarded as a flatness condition of the connection 
$\nabla$ on $\mu$.
We do not require the condition (M0) for $J$ in our paper.
\end{example}

\subsection{Lie $(n+1)$-algebroid and higher Dirac structure}\label{HDS}
The compatibility condition \eqref{momsec1} is regarded as the higher Dirac structure of a Lie $m$-algebroid induced from the Lie algebroid $E$.
Let $m=n+1$ in this section.

A Lie $m$-algebroid is a higher analogue of a Lie algebroid.
A QP-manifold description based on graded geometry provides a 
clear method of the definition of a Lie $m$-algebroid.
A graded manifold $(M, \calO_{\calM})$ is a ringed space 
whose structure sheaf $\calO_{M}$ is a $\bZ$-graded commutative algebra over an ordinary smooth manifold $M$. The grading is compatible with the supermanifold grading, that is, a variable of even degree is commutative and a variable of odd degree is anticommutative. By definition, the structure sheaf of $M$ is locally isomorphic to $C^{\infty}(U)\otimes S^{\bullet}(V)$, where $U$ is a local chart on $M$, $V$ is a graded vector space, and $S^{\bullet}(V)$ is a free graded commutative ring on $V$. 

We consider nonnegatively graded manifolds called an N-manifold 
in this section. Let $m$ be a nonnegative integer.
\begin{definition} 
Consider an N-manifold $\calM$ equipped with a graded symplectic structure $\omega$ of degree $m$, and a vector field $Q$ of degree plus one such that $Q^2=0$.
$\calM$ is called a \textit{QP-manifold} of degree $m$ if $(\omega, Q)$ satisfies $\calL_Q \omega =0$.
\cite{Schwarz:1992nx}
\end{definition} 
We call this vector field $Q$ the homological vector field.
The graded Poisson bracket $\sbv{-}{-}$ is given from symplectic form $\omega$.
For any QP-manifold, there exists a homological function $\Theta \in C^{\infty}((\calM)$ associated to $Q$ such that 
$Q = \sbv{\Theta}{-}$
The homological condition, $Q^2=0$, implies that $\Theta$ is a solution of 
Equation 
\begin{align}
\sbv{\Theta}{\Theta} =0.
\label{CME}
\end{align}
One can refer to some references of mathematics of a QP-manifold \cite{Roytenberg:2006qz, Ikeda:2012pv}.

Let $E$ be a vector bundle over $M$.
In order to construct a Lie m-algebroid induced from a vector bundle $E$ as a QP-manifold, we consider the $(m-1)$-shifted vector bundle, $E[m-1]$ and the $1$-shifted dual bundle $E^*[1]$ with coordinates of the fiber shifted by $m-1$ and $1$.

A function on a graded manifold $E^*[m-1] \oplus E[1]$ is identified to 
a section on the vector bundle $E \oplus \wedge^{m-1} E^*$.
$C^{\infty}(E^*[m-1] \oplus E[1]) \simeq \Gamma(E \oplus \wedge^{m-1} E^*)$.
Let me explain the precise correspondence of two spaces.
We consider a QP-manifold $\calM = E^*[m-1] \oplus E[1] \oplus T^*[m]M$.
Let $\partial_i$, $e_a$ and $e^a$ be the basis of $TM$, $E$ and $E^*$,
respectively. The map 
\begin{align}
j: E \oplus \wedge^{m-1} E^* \oplus TM 
\rightarrow E^*[m-1] \oplus E[1] \oplus T^*[m]M,
\label{defj}
\end{align}
is given by 
$j:(x^i, e_a, e^a, \partial_i) \mapsto (x^i, p_a, q^a, \xi_i)$,
where $p_a$, $q^a$ and $\xi_i$ are coordinates of $E^*[m-1]$, $E[1]$ 
and $T^*[m]M$ of degree $(1, m-1, m)$, respectively.
The map $j$ induces the map,
$j: C^{\infty}(E^*[m-1] \oplus E[1]) \rightarrow \Gamma(E \oplus \wedge^{m-1} E^*)$.
A canonical graded symplectic form is defined by
\begin{align}
\omega &= \delta x^i \wedge \delta \xi_i + \delta q^a \wedge \delta p_a,
\label{gradedss}
\end{align}
where $\delta$ is the graded de Rham differential.

Now suppose a QP-manifold structure on $\calM = E^*[m-1] \oplus E[1] \oplus T^*[m]M$. i.e., take the canonical symplectic form \eqref{gradedss} 
and a homological function $\Theta$ satisfying Equation \eqref{CME}.

A Lie m-algebroid on $E \oplus \wedge^{m-1} E^*$ consists of 
an algebra on $\Gamma(E \oplus \wedge^{m-1} E^*)$ over $C^{\infty}(M)$
with three operations,
$(\inner{-}{-}, \rho, \courant{-}{-})$. 
$\inner{-}{-}: 
\Gamma(E \oplus \wedge^{m-1} E^*) \otimes \Gamma(E \oplus \wedge^{m-1} E^*)
\rightarrow \Gamma(\wedge^{m-2} E^*)$
is an inner product.
The bundle map 
$\rho: E \oplus \wedge^{m-1} E^* \rightarrow TM$ is called the anchor map, 
and the bilinear bracket
$\courant{-}{-}: \Gamma(E \oplus \wedge^{m-1} E^*)
\times \Gamma(E \oplus \wedge^{m-1} E^*) \rightarrow \Gamma(E \oplus \wedge^{m-1} E^*)$ is called the (higher) Dorfman bracket.
In the QP-manifold description, they are defined by
\begin{align}
\inner{e_1}{e_2} &= j^* \sbv{\ue_1}{\ue_2},
\\
\rho(e)f &= j^* \sbv{\sbv{\ue}{\Theta}}{f},
\\
\courant{e_1}{e_2} &= j^* \sbv{\sbv{\ue_1}{\Theta}}{\ue_2},
\end{align}
for $e, e_1, e_2 \in \Gamma(E \oplus \wedge^{m-1} E^*)$ 
and $f \in C^{\infty}(M)$.
Here $\ue = j_* e$ is the super function corresponding to $e \in \Gamma(E)$.
$j_*$ and $j^*$ are the pushforward and the pullback with respect to the map $j$ defined in Equation \eqref{defj}.
All the identities of three operations are induced from one equation \eqref{CME}.
We identify the graded manifold description and the normal vector bundle description and drop the operation $j$.

Now let the vector bundle $E$ be a Lie algebroid and $M$ be an $m$-plectic manifold. Then $E$ has the anchor map and the Lie bracket $\rho, [-,-]$ and $M$ has a closed $(m+1)$-form $H$.
If we define 
\begin{align}
\Theta &= \Theta_0 + \underline{ \iota_{\rho}^{m+1} H} 
\nonumber \\ 
&= \rho^i_a(x) \xi_i q^a + \frac{1}{2} C_{bc}^a(x) p_a q^b q^c 
+ \frac{1}{(m+1)!} \rho^{i_1}_{a_1} \ldots \rho^{i_{m+1}}_{a_{m+1}} 
H_{i_1 \ldots i_{m+1}}(x) q^{a_1}\ldots q^{a_{m+1}},
\label{homological}
\end{align}
$\Theta$ satisfies $\sbv{\Theta}{\Theta}=0$,
where 
\begin{align}
\Theta_0 &= \rho^i_a(x) \xi_i q^a + \frac{1}{2} C_{bc}^a(x) p_a q^b q^c.
\label{homological1}
\end{align}
Because $\Theta_0$ satisfies $\sbv{\Theta_0}{\Theta_0}=0$ from the identities of the Lie algebroid and $\sbv{\Theta_0}{\underline{\iota_{\rho}^{m+1} H} }=0$ is given from $\rd H = 0$.
$\sbv{\underline{\iota_{\rho}^{m+1} H} }{\underline{\iota_{\rho}^{m+1} H}} =0$ is trivially satisfied.
Thus Equation \eqref{CME} is satisfied, and it gives a QP-manifold. 
Therefore Equation \eqref{homological} defines a Lie m-algebroid.
Note that Equation \eqref{homological} does not include $J$.
Equation ${}^E \rd J = - \bracket{\otimes^{m+1} \rho}{H}$, is described as
a higher Dirac structure, which is explained next.

Three operations of this Lie m-algebroid are as follows.
Let $u + \alpha, v + \beta \in \Gamma(E \oplus \wedge^{m-1} E^*)$,
where $u, v \in \Gamma(E)$ and 
$\alpha, \beta \in \Gamma(\wedge^{m-1} E^*)$.
\begin{align}
\inner{u + \alpha}{v + \beta} &= \inner{u}{\beta} + \inner{\alpha}{v},
\label{operation1}
\\
\rho(e)f &= \rho(u) f,
\label{operation2}
\\
\courant{u + \alpha }{v + \beta} &= [u, v] + \calL_{u} \beta - \iota_v {}^E \rd \alpha + \iota_u \iota_v (\iota_{\rho}^{m+1} H),
\label{operation3}
\end{align}
where the bracket $\inner{-}{-}$ in the right hand side of \eqref{operation1}
is the pairing of $E$ and $E^*$.
$\rho$ in the right hand side of \eqref{operation2} is the anchor map of the Lie algebroid $E$.
The interior product of the right hand side of 
\eqref{operation3} is the contraction with respect to $E$ and $E^*$,
The Lie derivative is $\calL_u = \iota_u {}^E \rd + {}^E \rd \iota_u$.

The higher Dirac structure is the subbundle $L$ of the Lie m-algebroid 
satisfying the conditions,
$\inner{e_1}{e_2} = 0$ for all $e_1, e_2 \in \Gamma(L)$, and 
$\courant{e_1}{e_2}$ is an element of $\Gamma(L)$,
i.e., $\Gamma(L)$ is involutive with respect to the bracket 
$\courant{-}{-}$.

Now we take $J \in \Gamma(\wedge^m E^*)$ satisfying Equation \eqref{momsec1}, 
i.e., ${}^E \rd J = - \bracket{\otimes^{m+1} \rho}{H}$.
Then we consider the set
\begin{align}
\Gamma(L) &= \{u + \inner{J}{u} \in \Gamma(E \oplus \wedge^{m-1} E^*)
| u \in \Gamma(E) \}.
\end{align}
\begin{theorem}
If $J$ and $H$ satisfy ${}^E \rd J = - \bracket{\otimes^{m+1} \rho}{H}$, 
$L$ is a higher Dirac structure of a Lie $m$-algebroid.
\end{theorem}
\begin{proof}
In fact, the inner product of two elements of $\Gamma(L)$, 
$u + \inner{J}{u}$ and $v + \inner{J}{v}$ for $u, v \in \Gamma(E)$
is 
\begin{align}
\inner{u + \inner{J}{u}}{v + \inner{J}{v}} &= 
\inner{u}{{J}{v}} + \inner{{J}{u}}{v} = 0,
\end{align}
from completely antisymmetricity of $J$.
Moreover the Dorfman bracket is computed by the derived bracket of the graded
functions,
\begin{align}
\courant{u + \inner{J}{u}}{v + \inner{J}{v}} &= 
j^* \sbv{\sbv{\underline{u + \inner{J}{u}}}{\Theta}}
{\underline{v + \inner{J}{v}}}
\nonumber \\ 
&= [u, v] + \inner{J}{[u, v]},
\end{align}
which is the element of $\Gamma(L)$ again.
Here we used $\sbv{\Theta_0}{j^* J} = - \iota_{\rho}^{m+1} H$
induced from ${}^E \rd J = - \bracket{\otimes^{m+1} \rho}{H}$.
\end{proof}

The Q-structure \eqref{homological} correctly gives the Lie algebroid structure with $H$, however there is no information of $J$.
The geometric condition, ${}^E \rd J = - \bracket{\otimes^{m+1} \rho}{H}$, is realized as a higher Dirac structure of a Lie $m$-algebroid induced by 
Equation \eqref{homological}.
A higher Dirac structure is not generally realized as a QP-manifold.
As a result, our Lie algebroid sigma model with WZ term cannot be formulated as an AKSZ sigma model since the AKSZ sigma model has a QP-manifold structure.
This is analogous to the twisted Poisson structure, which is a Dirac structure of the standard Courant algebroid on $TM \oplus T^*M$.
Though the standard Courant algebroid is a QP-manifold of degree 2, 
the twisted Poisson structure is not realized as any QP-manifold.
It is known that the twisted Poisson sigma model cannot be formulated as 
an AKSZ sigma model.

\section{Hamiltonian formalism}\label{Hamform}
In this section, 
the Hamiltonian formalism and constraints are analyzed
to make the action functional \eqref{classicalactionofHLASM0} consistent.
We show that the classical action \eqref{classicalactionofHLASM0} is consistent if the target space geometric data satisfy Equation \eqref{momsec1}, i.e., the target space is a pre-multisymplectic manifold with a Lie algebroid action and a bracket-compatible $E$-flux.

Take the worldvolume, $\Sigma = \bR \times T^n$ or $\Sigma = S^1 \times T^n$. 
Canonical conjugate momenta of $X^i$ and $A_i$ 
\footnote{$Z_i$ and $Y_a$ appear as canonical conjugates
of $X^i$ and $A_i$.}
are
\begin{eqnarray}
P_{X i} &=& \frac{\delta S}{\delta \dot{X}^i} 
= \left(Z_{i} + \frac{1}{n!} (-1)^{n} B_{ij_1\ldots j_n}(X) 
\rd X^{j_1} \wedge \ldots \wedge \rd X^{j_n}\right)^{(s)}
\nonumber \\ 
&=& \frac{1}{n!} \epsilon^{0 \mu_1 \ldots \mu_n} 
\left(Z_{\mu_1 \ldots \mu_n i} 
+ (-1)^{n} B_{ij_1\ldots j_n}(X) 
\partial_{\mu_1} X^{j_1} \wedge \ldots \wedge \partial_{\mu_n} X^{j_n}
\right),
\label{PB01}
\\
P_{A a}^{\mu} &=& \frac{\delta S}{\delta \dot{A_{\mu}}^a} 
= Y^{(s)}_a = \frac{1}{(n-1)!} \epsilon^{0 \mu \nu_2 \ldots \nu_n} Y
_{\nu_2 \ldots \nu_n a}.
\label{PB02}
\end{eqnarray}
where $(s)$ means coefficient functions of the space components of 
the differential forms on $\Sigma$. 
$\mu, \nu=1,\ldots,n$ are spatial indices on $\Sigma$
and $0$ is the time component.

Substituting Equations \eqref{PB01} and \eqref{PB02} to the basic Poisson bracket of canonical quantities,
$\{x^I(\sigma), p_J(\sigma^{\prime})\}_{PB} = \delta^I_J \delta(\sigma - \sigma^{\prime})$,
we obtain Poisson brackets of fundamental fields,
\begin{align}
\{\xzero^i(\sigma), Z^{(s)}_{j}(\sigma^{\prime})\}_{PB} &= \delta^i{}_j \delta(\sigma - \sigma^{\prime}),
\label{PB1}
\\
\{A^{(s)a}(\sigma), Y^{(s)}_b(\sigma^{\prime}) \}_{PB} &= \delta^a{}_b \delta(\sigma - \sigma^{\prime})
= (-1)^{n-1} \{Y^{(s)}_b(\sigma), A^{(s)a}(\sigma^{\prime}) \}_{PB},
\label{PB2}
\\
\{Z^{(s)}_{i}(\sigma), Z^{(s)}_{j}(\sigma^{\prime})\}_{PB} &=
\frac{(-1)^n}{2 n!} H_{ijk_1\cdots k_{n}}(\xzero(\sigma)) (\rd X^{k_1} 
\wedge \ldots \wedge \rd X^{k_{n}})^{(s)}\delta^n(\sigma - \sigma^{\prime}).
\label{PB3}
\end{align}
The symplectic form corresponding to these Poisson brackets 
\eqref{PB1}--\eqref{PB3} is
\begin{eqnarray}
\omega &=& \int_{T^n} 
\left(\delta X^i \wedge \delta Z^{(s)}_{i}
+ \delta A^{(s)a} \wedge \delta Y^{(s)}_a
\right.
\nonumber \\ && 
\left.
+ \frac{(-1)^{n+1}}{n!} H_{i_1\ldots i_n jk}(\xzero) (\rd \xzero^{i_1}
\wedge \ldots \wedge \rd \xzero^{i_n})^{(s)} \delta \xzero^j \wedge \delta \xzero^k
\right).
\end{eqnarray}

%
The canonical conjugates of time components $A^{(0) a}$, $Y^{(0)}_{a}$, 
$Z^{(0)}_{i}$ are $0$. These give primary constraints:
\begin{eqnarray}
P_{A^{(0)} a} &\approx& 0,
\qquad 
P_{Y^{(0)}}^a \approx 0,
\qquad
P_{Z^{(0)} i} \approx 0,
\label{primary3}
\end{eqnarray}
where $(0)$ denotes the time component of the field.
The Hamiltonian is proportional to constraints,
\begin{eqnarray}
{\cal H} = \int_{T^n} d^{n+1} \sigma
(Z_{0i} G_X^i + Y_{0a} G_A^a + A_0^a G_{Ya}).
\label{hamiltonian}
\end{eqnarray}
Here $G$'s are constraints without time derivatives,
\begin{align}
G^i_X &:= (\rd X^i - \rho^i_a(X) A^{a})^{(s)} \, ,
 \label{constraint1}\\
G_{A}^a &:= (\rd A^a + \tfrac{1}{2} C^a_{bc}(X) A^b \wedge A^c)^{(s)} \, ,
\label{constraint2}
\\
G_{Ya} &:= \left(\rd Y_a + (-1)^{n} \rho^i_a(X) Z_i
+ (-1)^{n-1} C_{ab}^c(X) Y_c \wedge A^b + \frac{1}{n!} J_{a b_2\ldots b_{n+1}}(\xzero) A^{b_2} \wedge \ldots \wedge 
A^{b_{n+1}} \right)^{(s)},
\label{constraint3}
\end{align}
which are spatial parts of equations of motion.
The secondary constraints are calculated by computing Poisson brackets with primary constraints \eqref{primary3} and the Hamiltonian $\calH$. The secondary constraints are $G^i_X$, $G_{A}^a$ and $G_{Ya}$,
\begin{align}
G^i_X & \approx 0, 
\qquad 
G_{A}^a \approx 0,
\qquad 
G_{Ya} \approx 0.
\end{align}

For the consistency condition of the mechanics, 
we require that $G^i_X$, $G_{A}^a$ and $G_{Ya}$ are first class constraints, i.e., Eqs.~\eqref{constraint1}--\eqref{constraint3} generate a closed algebra under Poisson brackets.

We suppose that a Lie algebroid structure on the target space vector bundle $E$. $\rho^i_a$ and $C_{ab}^c$ are local coordinate expressions of the anchor map and structure functions satisfying Equations \eqref{LAidentity1} and \eqref{LAidentity2}.
Moreover suppose that $H$ in the WZ term is a closed $(n+2)$-form.
Under the above assumptions, Poisson brackets of constraints $G^i_X$, $G_{A}^a$ and $G_{Ya}$ are computed using the fundamental Poisson brackets
\eqref{PB1}--\eqref{PB3}.
They are the first class if and only if $J$ satisfies the bracket-compatible condition \eqref{momsec1}.
In fact, under Equation \eqref{momsec1}, we obtain the following Poisson brackets of three constraints,
\begin{eqnarray}
&& \{G^i_X(\sigma), G^j_X(\sigma^{\prime}) \}_{PB} = 0,
\\
&& \{G^i_X(\sigma), G_{A}^a(\sigma^{\prime}) \}_{PB} = 0,
\\
&& \{G^i_X(\sigma), G_{Ya}(\sigma^{\prime}) \}_{PB} =
(-1)^{n-1} \partial_j \rho^i_a G^j_X(\sigma)
\delta^n(\sigma - \sigma^{\prime}),
\\
&& \{G_{A}^a(\sigma), G_{A}^b(\sigma^{\prime}) \}_{PB} = 0,
\\
&& \{G_{A}^a(\sigma), G_{Yb}(\sigma^{\prime}) \}_{PB} = 
(-1)^n [\partial_i C^a_{bc} A^{c} \wedge G_X^i(\sigma)
+ C^a_{bc} G_{A}^c(\sigma) ]^{(s)} \delta^n(\sigma - \sigma^{\prime}),
\\
&& \{G_{Ya}(\sigma), G_{Yb}(\sigma^{\prime}) \}_{PB} = 
\left[\left(\partial_i C_{ab}^c Y_c + \frac{(-1)^{n-1}}{n!} \partial_i 
J_{a b c_3\ldots c_{n+1}} A^{c_3} \wedge \ldots \wedge A^{c_{n+1}} \right) 
\wedge G_X^i
\right.
\nonumber \\ && \qquad
\left.
+ (-1)^{n-1} C_{ab}^c G_{Yc} 
+ \frac{(-1)^{n-2}}{(n-1)!} J_{a b c e_4\ldots e_{n+1}} A^{e_4} \wedge \ldots \wedge A^{e_{n+1}} \wedge G_A^c
\right.
\nonumber \\ && \qquad
\left.
+ \frac{(-1)^{n-1}}{(n+1)!} \sum_{m=1}^n \rho^i_a \rho^j_b
H_{ij k_1\ldots k_m k_{m+1} \ldots k_{n}} \rd X^{k_1}\wedge \ldots
\wedge \rd X^{k_{m-1}} \wedge G_X^{k_{m}} 
\right.
\nonumber \\ && \qquad
\left.
\wedge  \rho^{k_{m+1}}_{c_{m+1}} A^{c_{m+1}} \wedge 
\ldots \wedge \rho^{k_{n}}_{c_{n}} A^{c_{n}}  \right]^{(s)} (\sigma) 
\delta^n(\sigma - \sigma^{\prime}),
\label{PBconstraint}
\end{eqnarray}
which shows that all the constraints are the first class.
Here $\sigma^{\mu}, \sigma^{'\mu}$ are local coordinates on $T^n$ and
all the fields are spatial components.
Equation \eqref{momsec1} is necessary for closedness of the final Poisson bracket \eqref{PBconstraint}.
The detail computation of Equation \eqref{PBconstraint} appears in Appendix \ref{CEq}.

\section{Target space covariantization}\label{covariantconstraints}
Constraints and Poisson brackets are rewritten by geometric quantities of the target Lie algebroid by introducing a connection $\nabla$ on $E$.

Let $\omega = \omega^b_{ai} dx^i \otimes e^a \otimes e_b$ be 
the connection $1$-form for the connection $\nabla$. 
%
Let $s, s' \in \Gamma(E)$.
Additional to the following ordinary curvature,
\beqa
R(s, s^{\prime}) &:=& [\nabla_s, \nabla_{s^{\prime}}] - \nabla_{[s, s^{\prime}]}, 
\eeqa
in a Lie algebroid, the following $E$-torsion $T$, the $E$-curvature and the basic curvature $S$ are defined, \cite{Blaom, Kotov:2016lpx}
\beqa
T(s, s^{\prime}) &:=& {}^E\nabla_s s^{\prime} - {}^E\nabla_{s^{\prime}} s
- [s, s^{\prime}],
\\
{}^ER(s, s^{\prime}) &:=& [{}^E\nabla_s, {}^E\nabla_{s^{\prime}}] - {}^E\nabla_{[s, s^{\prime}]}, 
\\
S(s, s^{\prime}) &:=& \calL_s (\nabla s^{\prime}) 
- \calL_{s^{\prime}} (\nabla s) 
- \nabla_{\rho(\nabla s)} s^{\prime} + \nabla_{\rho(\nabla s^{\prime})} s
- \nabla[s, s^{\prime}] 
\nonumber \\ &&
= (\nabla T + 2 \mathrm{Alt} \, \iota_\rho R)(s, s^{\prime}).
\eeqa
Local coordinate expressions appear in Appendix \ref{geometryofLA}.

We can rewrite constraints as follows.
Since $G^i_X$ is already covariant under the target space diffeomorphism,
the local coordinate expression is the same as Equation \eqref{constraint1}.
$G_{A}^a$ and $G_{Ya}$ are written as
\begin{align}
G_{A}^{\nabla a} &:= (\nabla A^a - \tfrac{1}{2} T^a_{bc}(X) A^b \wedge A^c)^{(s)} \, ,
\label{constraint12}
\\
G^{\nabla}_{Ya} &:= \left(\nabla Y_a + (-1)^{n} \rho^i_a(X) Z_i
+ (-1)^{n} T_{ab}^c(X) Y_c \wedge A^b 
\right.
\nonumber \\ & \qquad
\left.
+ \frac{1}{n!} J_{a b_2\ldots b_{n+1}}(\xzero) A^{b_2} \wedge \ldots \wedge 
A^{b_{n+1}}
\right)^{(s)},
\label{constraint13}
\end{align}
where 
\begin{align}
\nabla A^a &:= \rd A^a - \omega_{bi}^a A^b \rd X^i,
\\
\nabla Y_a &:= \rd Y_a + (-1)^{n} \omega_{ai}^c Y_c \rd X^i,
\end{align}
and the covariantized constraints are given by
\begin{align}
G_{A}^{\nabla a} &= G_{A}^{a} - (\omega_{bi}^a(X) A^b G_X^i)^{(s)} \, ,
\\
G^{\nabla}_{Ya} &= G_{Ya} 
+ \left( (-1)^{n} \omega_{ai}^c Y_c G_X^i \right)^{(s)}.
\end{align}
If we impose the bracket-compatible condition \eqref{momsec1},
we obtain the following Poisson brackets,
\begin{eqnarray}
&& \{G^i_X(\sigma), G^j_X(\sigma^{\prime}) \}_{PB} = 0,
\label{covPB1}
\\
&& \{G^i_X(\sigma), G_{A}^{\nabla a}(\sigma^{\prime}) \}_{PB} = 0,
\\
&& \{G^i_X(\sigma), G_{Ya}^{\nabla}(\sigma^{\prime}) \}_{PB} =
(-1)^{n-1} \nabla_j \rho^i_a G^j_X(\sigma)
\delta^n(\sigma - \sigma^{\prime}),
\\
&& \{G_{A}^{\nabla a}(\sigma), G_{A}^{\nabla b}(\sigma^{\prime}) \}_{PB} = 0,
\\
&& \{G_{A}^{\nabla a}(\sigma), G_{Yb}^{\nabla}(\sigma^{\prime}) \}_{PB} = 
(-1)^{n+1} [S_{ibc}^a A^{c} \wedge G_X^i(\sigma)
+ T^a_{bc} G_{A}^{\nabla c}(\sigma) ]^{(s)} \delta^n(\sigma - \sigma^{\prime}),
\\
&& \{G_{Ya}^{\nabla}(\sigma), G_{Yb}^{\nabla}(\sigma^{\prime}) \}_{PB} = 
\left[\left(- S_{iab}^c Y_c + \frac{(-1)^{n-1}}{n!} \nabla_i 
J_{a b c_3\ldots c_{n+1}} A^{c_3} \wedge \ldots \wedge A^{c_{n+1}} \right) 
\wedge G_X^i
\right.
\nonumber \\ && \qquad
\left.
+ (-1)^{n} T_{ab}^c G_{Yc}^{\nabla} 
+ \frac{(-1)^{n-2}}{(n-1)!} J_{a b c e_4\ldots e_{n+1}} A^{e_4} \wedge \ldots \wedge A^{e_{n+1}} \wedge G_A^{\nabla c}
\right.
\nonumber \\ && \qquad
\left.
+ \frac{(-1)^{n-1}}{(n+1)!} \sum_{m=1}^n \rho^i_a \rho^j_b
H_{ij k_1\ldots k_m k_{m+1} \ldots k_{n}} \rd X^{k_1}\wedge \ldots
\wedge \rd X^{k_{m-1}} \wedge G_X^{k_{m}} 
\right.
\nonumber \\ && \qquad
\left.
\wedge  \rho^{k_{m+1}}_{c_{m+1}} A^{c_{m+1}} \wedge 
\ldots \wedge \rho^{k_{n}}_{c_{n}} A^{c_{n}}  \right]^{(s)} (\sigma) 
\delta^n(\sigma - \sigma^{\prime}),
\label{covPB2}
\end{eqnarray}
which shows all the constraints are the first class.
Here
$\nabla_i \rho^j_a = \partial_i \rho^j_a + \omega_{ai}^b \rho^j_b$.
The coefficients of Poisson brackets are written by $\rho$, $H$, $J$, $\nabla$, $T$ and $S$.
Therefore we obtain the following result
from Equations \eqref{covPB1}--\eqref{covPB2}.
\begin{theorem}
Suppose that the target space has a Lie algebroid structure and $\rd H=0$. 
Then, constraints $G^i_X$, $G_{A}^a$ and $G_{Ya}$ are the first class if and only if $J$ satisfies the bracket-compatible condition \eqref{momsec1}.
\end{theorem}

\section{Gauge transformation}
In this section, we discuss the Lagrangian formalism.

The gauge transformations are given from constraints and Poisson brackets
in the Hamiltonian formalism.
From the general theory of the analytical mechanics, a gauge transformation of
a field $\Phi$ in the Lagrangian formalism is computed by the Poisson bracket 
of constraints and $\Phi$,
\begin{eqnarray}
\delta \Phi &=& \left\{\epsilon^a G_a, \Phi \right\} + \tau^a(\Phi) G_a,
\end{eqnarray}
where we should carefully fix freedom of the term $\tau^a(\Phi) G_a$,
which is the freedom of on-shell vanishing trivial gauge transformations.
$\tau^a(\Phi)$ is an arbitrary function of fields.
These ambiguities and problems were discussed in the paper \cite{Ikeda:2020eft} for the twisted Poisson sigma model.
In the twisted Poisson sigma model, $\tau^a(\Phi)$ is a nonzero function.
The situation for our twisted Lie algebroid topological sigma model
is similar to the twisted Poisson sigma model.
We need a nontrivial term $\tau^a(\Phi)$ and it is fixed by imposing 
the Lorentz, or diffeomorphism covariance of gauge transformations on $\Sigma$.

Using this formula, we can compute gauge transformations of each field from constraints in Section \ref{Hamform}. 
We need three gauge parameters corresponding to constraints $G_{Ya}$, $G_A^a$
and $G_X^i$, 
$c^a \in \Gamma(\Sigma, X^*(E))$,
$t_a \in \Gamma(\wedge^{n-2} T^*\Sigma, X^*(E^*))$,
$w_i \in \Gamma(\wedge^{n-1} T^*\Sigma, X^*(T^*M))$.
$c^a$ is a function, $t_a$ is an $(n-2)$-form and 
$w_i$ is an $(n-1)$-form.

Gauge transformations of fundamental fields are given by
\begin{align}
\delta X^i &= \rho^i_a(X) c^a,
\label{gauge01}
\\
\delta A^a &= \rd c^a + C^a_{bc}(X) A^b c^c,
\label{gauge02}
\\
\delta Y_a &= \rd t_a + (-1)^{n} \rho^i_a(X) w_i
+ C_{ab}^c(X) (-Y_c c^b + (-1)^{n} t_c \wedge A^b)
\nonumber \\
& 
+ \frac{(-1)^n}{(n-1)!} J_{a b_2\ldots b_{n+1}}(X) A^{b_2} \wedge \ldots \wedge A^{b_n} c^{b_{n+1}},
\label{gauge03}
\\
\delta Z_i &= \rd w_i 
+ \partial_i \rho^j_a (- Z_j \wedge c^a + (-1)^n w_j \wedge A^a)
+ \frac{1}{2} \partial_i C^a_{bc}
(2 Y_a \wedge A^b c^c + (-1)^{n} t_a \wedge A^b \wedge A^c)
\nonumber \\
& + \frac{1}{n!} \partial_i 
J_{a_1 \ldots a_{n+1}}(X) A^{a_1} \wedge \ldots \wedge A^{a_n} c^{a_{n+1}}
\nonumber \\
&- \frac{1}{(n+1)!} H_{ij_1\ldots j_n k} 
\sum_{m=0}^{n} \rd X^{j_1} \wedge \ldots \wedge \rd X^{j_m}
\wedge \rho^{j_{m+1}}_{a_{m+1}} A^{a_{m+1}} \wedge \ldots 
\wedge \rho^{j_{n}}_{a_{n}} A^{a_{n}} 
\rho^k_b c^b
\nonumber \\ 
&= \rd w_i 
+ \partial_i \rho^j_a (- Z_j \wedge c^a + (-1)^n w_j \wedge A^a)
+ \frac{1}{2} \partial_i C^a_{bc}
(2 Y_a \wedge A^b c^c + (-1)^{n} t_a \wedge A^b \wedge A^c)
\nonumber \\
& + \frac{1}{n!} \left(\partial_i J_{a_1 \ldots a_{n+1}}(X) 
- \rho^{j_{1}}_{a_{1}} \ldots \rho^{j_{n}}_{a_{n}} 
\rho^k_{a_{n+1}} H_{ij_1\ldots j_n k} 
\right) A^{a_1} \wedge \ldots \wedge A^{a_n} c^{a_{n+1}}
- \frac{1}{(n+1)!} H_{ij_1\ldots j_n k} 
\nonumber \\
& \times 
\sum_{m=1}^{n} (n-m+1) \rd X^{j_1} \wedge \ldots \wedge \rd X^{j_{m-1}}
\wedge F_X^{j_m} \wedge 
\rho^{j_{m+1}}_{a_{m+1}} A^{a_{m+1}} \wedge \ldots 
\wedge \rho^{j_{n}}_{a_{n}} A^{a_{n}} 
\rho^k_b c^b \, .
\label{gauge04}
\end{align}
In fact, the action functional \eqref{classicalactionofHLASM} 
is gauge invariant 
$\delta S=0$ under these gauge transformations 
\eqref{gauge01}--\eqref{gauge04}.

Computations of the gauge algebra are rather complicated, however
from the general theory of the analytical mechanics,
the first class constraints in the Hamiltonian formalism
give an on-shell closed gauge algebra such that 
$[\delta_1, \delta_2] \approx \delta_3$ in the Lagrangian formalism.

\section{Manifestly target space covariant gauge transformation}
Gauge transformations are written to target space covariant ones by 
introducing a connection $\nabla$ on $E$ as in Section
\ref{covariantconstraints}.
In gauge transformations of the basis of $E$ and $E^*$,
terms using the connection $1$-form $\omega_{ai}^b$ appear as 
follows,
\begin{align}
\delta^{\nabla} e_a &= \omega_{ai}^b(X) \delta X^i e_b
= \omega_{ai}^b(X) \rho^i_c c^c e_b,
\label{gfofbasis1}
\\
\delta^{\nabla} e^a 
&= - \omega_{bi}^a(X) \delta X^i e^b
= - \omega_{bi}^a(X) \rho^i_c c^c e^b.
\label{gfofbasis2}
\end{align}

The gauge transformation of $X^i$, Equation \eqref{gauge01}, is already covariant $\delta^{\nabla} X^i = \delta X^i$.
The covariant gauge transformation of $A^a$ is
\begin{eqnarray}
\delta^{\nabla} A^a 
&=& \rd c^a + C^a_{bc}(X) A^b c^c + \omega^a_{bi}(X) c^b F_X^i.
\label{gauge22}
\end{eqnarray}
In fact, using transformations of basis \eqref{gfofbasis1},
the gauge transformation of the coordinate independent 
form $A=A^a \otimes e_a$ is calculated as follows:
\begin{align}
\delta^{\nabla} A &= \delta^{\nabla} (A^a \otimes e_a)
\nonumber \\
&= \delta^{\nabla} A^a \otimes e_a + A^a \otimes \delta^{\nabla} e_a
\nonumber \\
&= (\rd c^a + C^a_{bc}(X) A^b c^c 
+ \omega^a_{bi}(X) c^b F_X^i + \rho^i_b \omega^a_{ci}(X) A^c c^b) \otimes e_a
\nonumber \\
&= (\nabla c^a - T^a_{bc}(X) A^b c^c) \otimes e_a
\nonumber \\ 
&= \nabla c - X^* T(A,c),
\label{covtransA}
\end{align}
where $\nabla c^a = \rd c^a + \omega_{bi}^a \rd X^i c^b$.
Equation \eqref{covtransA} is covariant under the diffeomorphism on $M$ and coordinate transformations on the fiber of $E$. For instance, 
$\omega_{ai}^b$ is transformed as
$\omega_{ai}^{'b} \rd x^i= (M^{-1})^c_a \omega_{ci}^{d} \rd x^i M_d^b 
+ (M^{-1})^c_a \rd M_c^b$
under a transition function $M_a^b(X)$ of the bundle,
and $A^a$ is transformed as $A^{\prime a} = M_b^a(X) A^b$, etc. 
Using transformations of all fields and functions under the 
transition function $M_a^b(X)$, 
we can check $\delta^{\nabla} A$ is invariant under coordinate changes
on the target vector bundle.

For $Y$, a similar calculation gives the following covariant gauge transformation,
\begin{eqnarray}
\delta^{\nabla} Y_a &=& \rd t_a + (-1)^n \rho^i_a(X) w_i
+ C_{ab}^c(X) (-Y_c c^b + A^b t_c)
\nonumber \\
&& 
+ \frac{1}{(n-1)!} J_{a b_2\ldots b_{n+1}}(X) A^{b_2} \wedge \ldots \wedge A^{b_n} c^{b_{n+1}}
+ (-1)^{n-1} \omega_{ai}^b F_X^i t_b.
\label{gauge23}
\end{eqnarray}
We can check the coordinate independent covariant gauge transformation,
\begin{align}
\delta^{\nabla} Y 
&= \delta^{\nabla} (Y_a \otimes e^a)
\nonumber \\
&= \left(\nabla t_a + (-1)^n \rho^i_a(X) w_i^{\nabla} 
- T_{ab}^c(X) (-Y_c c^b + A^b t_c)
\right.
\nonumber \\ & 
\left.
+ \frac{1}{n!} J_{a b_1\ldots b_{n+1}}(X) A^{b_1} \wedge \ldots \wedge A^{b_n} c^{b_{n+1}} \right) \otimes e^a,
\nonumber \\
&= \nabla t + (-1)^n \iota_{X^* \rho} w^{\nabla} 
+ X^* T (Y, c) - X^* T(A, t)
+ X^* J (A, \ldots, A, c),
\label{covtransY}
\end{align}
where
\begin{eqnarray}
\nabla t_a &=& \rd t_a - \omega_{ai}^b \rd X^i t_b,
\\
w_i^{\nabla} &=& w_i + (-1)^{n-1} \omega_{bi}^c (-Y_c c^b + A^b t_c).
\end{eqnarray}
Similarly, we obtain the covariant gauge transformation of $Z$ as
\begin{eqnarray}
\delta^{\nabla} Z_i 
&=& \nabla w_i^{\nabla} 
+ \nabla_i \rho^j_a (-Z_j \wedge c^a + (-1)^n w_j^{\nabla}  \wedge A^a)
- \tfrac{1}{2} S^a_{ibc}
(2 Y_a \wedge A^b c^c + (-1)^{n} A^b \wedge A^c t_a)
\nonumber \\
&& + \frac{1}{n!} \left(\nabla_i J_{a_1 \ldots a_{n+1}}(X) 
- \rho^{j_{1}}_{a_{1}} \ldots \rho^{j_{n}}_{a_{n}} \rho^{j_{n+1}}_{a_{n+1}} H_{ij_1\ldots j_{n+1}} 
\right) A^{a_1} \wedge \ldots \wedge A^{a_n} c^{a_{n+1}}
- \frac{1}{(n+1)!} H_{ij_1\ldots j_n k} 
\nonumber \\
&& \times 
\sum_{m=1}^{n} (n-m+1) \rd X^{j_1} \wedge \ldots \wedge \rd X^{j_{m-1}}
\wedge F_X^{j_m} \wedge 
\rho^{j_{m+1}}_{a_{m+1}} A^{a_{m+1}} \wedge \ldots 
\wedge \rho^{j_{n}}_{a_{n}} A^{a_{n}} 
\rho^k_b c^b \, .
\end{eqnarray}
The coordinate independent form is 
\begin{eqnarray}
\delta^{\nabla} Z
&=& \nabla w^{\nabla} 
- \iota_{X^* \nabla \rho(c)} Z + \iota_{X^* \nabla \rho(A)} w^{\nabla}
- X^* S(Y, A, c) + (-1)^{n} X^* S(t, A, A)
\nonumber \\
&& + X^* \nabla J(A, \ldots, A, c) 
- 
\iota_{X^* \rho(c) }\iota_{X^* \rho(A)}^{n} H
\nonumber \\
&& 
+ \sum_{m=1}^{n} (n-m+1) 
(-1)^{n} \iota_{X^* \rho (c)} \iota_{F_X} \iota_{X^* \rho (A)}^{(n-m)} H \, .
\label{covtransZ}
\end{eqnarray}
We obtain invariant coordinate independent gauge transformations \eqref{covtransA}, \eqref{covtransY} and \eqref{covtransZ}.

\section{Conclusion and discussion}
We have constructed an $(n+1)$ dimensional topological sigma model with a Lie algebroid structure, an $E$-flux and the WZ term, generalizing the twisted Poisson sigma model and the twisted $R$-Poisson sigma model. The Poisson manifold target space is generalized to a Lie algebroid target space.
Moreover, from the consistency condition of constraints, we fixed a consistency condition of the $E$-flux, the WZ term and other coefficient functions. They are universal geometric conditions of compatibility of $E$-differential forms with a pre-multisymplectic structure under a Lie algebroid action. 
We pointed out that they were regarded as a Lie algebroid generalization of parts of the momentum map theory on the multi-symplectic manifold.
We will be able to understand and apply this result to geometric description of higher fluxes and dualities in higher dimensions.

In general, a higher dimensional topological sigma model of AKSZ type has a higher $L_\infty$-algebroid structure. If we deform the theory adding the WZ term to the action, the AKSZ construction does not work. We need to modify the AKSZ construction of the BV formalism for topological sigma models with the WZ term. 
Though the BFV and BV formalisms of the two dimensional twisted Poisson sigma model were geometrically constructed \cite{Ikeda:2019czt}, they are still open 
in higher dimensional topological sigma models with WZ term.
In order to construct the BFV and BV in higher dimensions, geometric analysis of compatibility conditions of the Lie $n$-, or $L_\infty$-algebroid structure with the pre-multisymplectic structure may be a key point.
The result in this paper gives a new insight and is one step.
The construction of the BV and BFV formalism of the twisted Lie algebroid sigma model and the twisted Lie-$n$, or $L_\infty$-algebroid sigma model are an important future problem for analysis of higher dimensional duality physics.

\subsection*{Acknowledgments}
\noindent
The author would like to thank Athanasios Chatzistavrakidis and
Yuji Hirota and  for useful comments and discussion.
This work was supported by the research promotion program for acquiring grants in-aid for scientific research(KAKENHI) in Ritsumeikan university.

\appendix
\section{Geometry of Lie algebroid}\label{geometryofLA}
We summarize notation, formulas and their local coordinate expressions of geometry of a Lie algebroid.

Let $(E, \rho, [-,-])$ be a Lie algebroid over a smooth manifold $M$.
$x^i$ is a local coordinate of $M$ and $e_a \in \Gamma(E)$ is a basis of sections of $E$. 
$i,j$, etc. are indices on $M$ and $a,b$, etc. are indices on the fiber of $E$.
Local coordinate expressions of the anchor map and the Lie bracket are
$\rho(e_a) f = \rho^i_a(x) \partial_i f$,
$[e_a, e_b ] = C_{ab}^c(x) e_c$, where $\partial_i = \tfrac{\partial}{\partial x^i}$.
Then, the conditions of $\rho$ and $C$ are
\beqa 
&& \rho_a^j \partial_j \rho_{b}^i - \rho_b^j \partial_j \rho_{a}^i = C_{ab}^c \rho_c^i,
\label{LAidentity1}
\\
&& C_{ad}^e C_{bc}^d + \rho_a^i \partial_i C_{bc}^e + \mbox{Cycl}(abc) = 0.
\label{LAidentity2}
\eeqa

Let $\nabla$ be an ordinary connection on the vector bundle $E$.
An \textit{$E$-connection} ${}^E \nabla: \Gamma(TM) \rightarrow \Gamma(TM \otimes E^*) $ on the space of sections $\Gamma(TM)$ is defined by
\begin{eqnarray}
{}^E \nabla_{e} v &:=& \calL_{\rho(e)} v + \rho(\nabla_v e)
= [\rho(e), v] + \rho(\nabla_v e),
\end{eqnarray}
where 
$e \in \Gamma(E)$ and $v \in \Gamma(TM)$.
%
The $E$-connection ${}^E \nabla$ satisfies 
\begin{eqnarray}
{}^E \rd \bracket{v}{\alpha} = 
\bracket{{}^E \nabla e}{\alpha} + \bracket{e}{{}^E \nabla \alpha},
\end{eqnarray}
for a vector field $v$ and a $1$-form $\alpha$.
For a $1$-form $\alpha$, it is given by
\begin{eqnarray}
{}^E \nabla_{e} \alpha &:=& \calL_{\rho(e)} \alpha 
+ \bracket{\rho(\nabla e)}{\alpha}.
\end{eqnarray}
$\omega = \omega^b_{ai} dx^i \otimes e^a \otimes e_b$ be 
a connection $1$-form. Then, local coordinate expressions of 
covariant derivatives and the $E$-covariant derivative are
\begin{eqnarray}
\nabla_i \alpha^a &=& \partial_i \alpha^a - \omega_{bi}^a \alpha^b,
\\
\nabla_i \beta_a &=& \partial_i \beta_a + \omega_{ai}^b \beta_b,
\\
{}^E \nabla_a v^i
&=&  \rho_a^j \partial_j v^i - \partial_j \rho^i_a v^j
- \rho^i_b \omega^b_{aj} v^j,
\\
{}^E \nabla_a \alpha_i
&=&  \rho_a^j \partial_j \alpha_i + \partial_i \rho^j_a \alpha_j
+ \rho^j_b \omega^b_{ai} \alpha_j.
\end{eqnarray}
An E-torsion, a curvature, an $E$-curvature and a basic curvature,
$T$, $R$, ${}^ER$ and $S$ are defined by
\beqa
R(s, s^{\prime}) &:=& [\nabla_s, \nabla_{s^{\prime}}] - \nabla_{[s, s^{\prime}]}, 
\nonumber \\
T(s, s^{\prime}) &:=& {}^E\nabla_s s^{\prime} - {}^E\nabla_{s^{\prime}} s
- [s, s^{\prime}],
\nonumber \\
{}^ER(s, s^{\prime}) &:=& [{}^E\nabla_s, {}^E\nabla_{s^{\prime}}] - {}^E\nabla_{[s, s^{\prime}]}, 
\nonumber \\
S(s, s^{\prime}) &:=& \calL_s (\nabla s^{\prime}) 
- \calL_{s^{\prime}} (\nabla s) 
- \nabla_{\rho(\nabla s)} s^{\prime} + \nabla_{\rho(\nabla s^{\prime})} s
\nonumber \\ &&
- \nabla[s, s^{\prime}] = (\nabla T + 2 \mathrm{Alt} \, \iota_\rho R)(s, s^{\prime}),
\nonumber
\eeqa
The following local coordinate expressions are given as
\beqa 
T_{ab}^c &\equiv& 
- C_{ab}^c + \rho_a^i \omega_{bi}^c - \rho_b^i \omega_{ai}^c,
\\
R_{ijb}^a &\equiv& 
\partial_i \omega_{aj}^b - \partial_j \omega_{ai}^b 
+ \omega_{ai}^c \omega_{cj}^b - \omega_{aj}^c \omega_{ci}^b,
\\
S_{iab}^{c} &\equiv& 
\nabla_i T_{ab}^c + \rho_b^j R_{ija}^c - \rho_a^j R_{ijb}^c,
 \nonumber \\
&=& - \partial_i C^c_{ab} + \omega_{di}^c C_{ab}^d - \omega_{ai}^d C_{db}^c - \omega_{bi}^d C_{ad}^c
+ \rho_a^j \partial_j \omega_{bi}^c
- \rho_b^j \partial_j \omega_{ai}^c
\nonumber \\ && 
+ \partial_i \rho_a^j \omega_{bj}^c
- \partial_i \rho_b^j \omega_{aj}^c
+ \omega_{ai}^d \rho_d^j \omega_{bj} ^c
- \omega_{bi}^d \rho_d^j \omega_{aj} ^c,
\eeqa
where the covariant derivative $\nabla_i T_{ab}^c$ is
\beqa 
\nabla_i T_{ab}^c &\equiv& 
\partial_i T_{ab}^c
- \omega_{di}^c T_{ab}^d + \omega_{ai}^d T_{db}^c + \omega_{bi}^d T_{ad}^c.
\eeqa
the $E$-curvature is given from the basic curvature as
\beqa 
{}^E R_{abc}^d &=& \rho_c^i S_{iab}^d.
\eeqa

\section{Computation of Equation \eqref{PBconstraint}
}\label{CEq}
The derivation of Poisson bracket,
$\{G_{Ya}(\sigma), G_{Yb}(\sigma^{\prime}) \}_{PB}$
is slightly complicated.
Using Lie algebroid identities of $\rho^i_a$ and $C^a_{bc}$, 
identities of $J$ and $H$, we obtain
\begin{eqnarray}
&& \{G_{Ya}(\sigma), G_{Yb}(\sigma^{\prime}) \}_{PB} = 
\left[\left(\partial_i C_{ab}^c Y_c + \frac{(-1)^{n-1}}{n!} \partial_i 
J_{a b c_3\ldots c_{n+1}} A^{c_3} \wedge \ldots \wedge A^{c_{n+1}} \right) 
\wedge G_X^i
\right.
\nonumber \\ && \qquad
\left.
+ (-1)^{n-1} C_{ab}^c G_{Yc} 
+ \frac{(-1)^{n-2}}{(n-1)!} J_{a b c e_4\ldots e_{n+1}} A^{e_4} \wedge \ldots \wedge A^{e_{n+1}} \wedge G_A^c \right]^{(s)} (\sigma) 
\delta^n(\sigma - \sigma^{\prime})
\nonumber \\ && \qquad
+ \left[
\frac{(-1)^{n-1}}{(n+1)!} \rho^i_a \rho^j_b
H_{ij k_1 \ldots k_{n}} \rd X^{k_1}
\wedge \ldots \wedge \rd X^{k_{n}} 
\right.
\nonumber \\ && \qquad
\left.
+ \frac{(-1)^{n}}{(n+1)!} \rho^i_a \rho^j_b
H_{ij k_1 \ldots k_{n}} \rho^{k_{1}}_{c_{1}} A^{c_{1}} \wedge 
\ldots \wedge \rho^{k_{n}}_{c_{n}} A^{c_{n}} 
\right]^{(s)} (\sigma) 
\delta^n(\sigma - \sigma^{\prime}).
\label{PBconstraint2}
\end{eqnarray}
The final two terms are rewritten using the constraint $G_X^k$ as
\begin{eqnarray}
&& \left[
\frac{(-1)^{n-1}}{(n+1)!} \sum_{m=1}^n \rho^i_a \rho^j_b
H_{ij k_1\ldots k_m k_{m+1} \ldots k_{n}} \rd X^{k_1}\wedge \ldots
\wedge \rd X^{k_{m-1}} \wedge G_X^{k_{m}} 
\right.
\nonumber \\ && \qquad
\left.
\wedge  \rho^{k_{m+1}}_{c_{m+1}} A^{c_{m+1}} \wedge 
\ldots \wedge \rho^{k_{n}}_{c_{n}} A^{c_{n}} 
\right]^{(s)} (\sigma) 
\delta^n(\sigma - \sigma^{\prime}),
\end{eqnarray}
which gives Equation \eqref{PBconstraint}.

\newcommand{\bibit}{\sl}



\end{document}